\documentclass[10pt,twocolumn,twoside]{JCNtran}
\newif\ifPDF
\ifx\PDFoutput\undefined
  \PDFfalse
\else
  \ifnum\PDFoutput=1
    \PDFtrue
  \else
    \PDFfalse
  \fi
\fi

\usepackage[dvips]{epsfig}
\usepackage[T1]{fontenc}
\usepackage{amsmath}
\usepackage{amsfonts}   % if you want the fonts
\usepackage{amssymb}
\usepackage{aeguill}
\usepackage{cite}
\usepackage[dvips]{graphicx}
\usepackage{lineno}
\usepackage{hyperref}
\usepackage{balance}
\usepackage{multirow}
\makeatletter
\def\ScaleIfNeeded{%
\ifdim\Gin@nat@width>\linewidth
\linewidth
\else
\Gin@nat@width
\fi
}

\def\BibTeX{{\rm B\kern-.05em{\sc i\kern-.025em b}\kern-.08em
    T\kern-.1667em\lower.7ex\hbox{E}\kern-.125emX}}

\newtheorem{theorem}{Theorem}

\begin{document}
%\linenumbers 

\title{Spectrum Sharing-based Multi-hop Decode-and-Forward Relay Networks under Interference Constraints: Performance Analysis and Relay Position Optimization}
\author{Vo Nguyen Quoc Bao, Tran Thien Thanh, Tuan Duc Nguyen, Thanh Dinh Vu
\thanks{V. N. Q. Bao is with Posts and Telecommunications Institute of Technology (PTIT), Ho Chi Minh City, Vietnam, email: baovnq@ptithcm.edu.vn.}
\thanks{T. T. Thanh and T. D. Vu are with Ho Chi Minh City University of Technology (HCMUT), Vietnam National University Ho Chi Minh City (VNUHCM).}
\thanks{T. D. Nguyen is with International University (IU), Vietnam National University Ho Chi Minh City (VNUHCM).}
} \markboth{THE JOURNAL OF COMMUNICATIONS AND
NETWORKS, VOL. X, NO. X, XXX. 200X}{Vo Nguyen Quoc Bao
\lowercase{\textit{et al}}.: Spectrum Sharing-based Multi-hop Decode-and-Forward Relay Networks}
\maketitle

\begin{abstract}
%\boldmath
The exact closed-form expressions for outage probability and bit error rate of spectrum sharing-based multi-hop decode-and-forward (DF) relay networks in non-identical Rayleigh fading channels are derived. We also provide the approximate closed-form expression for the system ergodic capacity. Utilizing these tractable analytical formulas, we can study the impact of key network parameters on the performance of cognitive multi-hop relay networks under interference constraints. Using a linear network model, we derive an optimum relay position scheme by numerically solving an optimization problem of balancing average signal-to-noise ratio (SNR) of each hop. The numerical results show that the optimal scheme leads to SNR performance gains of more than 1 dB. All the analytical expressions are verified by Monte-Carlo simulations confirming the advantage of multihop DF relaying networks in cognitive environments. 
\end{abstract}


\section{Introduction}
Cognitive radio (CR) has drawn considerable attention in
the academic and  industrial communities in the past few years and
has been considered as one of the main feature for future wireless
networks \cite{Wang2011,Filin2011}. As an evolution of software-defined radio (SDR), CR with
ability of learning  from its surroundings and adapting its
transmitting configurations allows secondary (unlicensed) users to
opportunistically transmits data in bands licensed to primary users \cite{Mitola1999,Haykin2005,Pawelczak2011}.
As a result, it can alleviate the problem of spectrum congestion and
thus allowing for a more efficient spectrum utilization \cite{Goldsmith2009}.

Recently, CR has also been considered as the radio platform for
relaying networks. Previous works on cognitive networks have assumed
two types of cognitive operational modes including opportunistic
spectrum access (OSA) and spectrum sharing (SS) \cite{Goldsmith2009,Qing2007}. In the OSA
approach, CR users are allowed to transmit over the same frequency
band licensed to primary users (PUs) only when the frequency band is detected
vacant \cite{Lee2007,Suraweera2009,Bao2010}. For the latter case, secondary users (SUs) consisting of the source and relays may take advantage of a PU's frequency band by opportunistically
transmitting with high power as long as the PU's transmission is
strictly protected \cite{Ban2009,Lee2011}. In spite of the burden on PUs, the SS approach
can improve spectral efficiency more aggressively than the OSA
approach.

Recently, the performance analysis for SS-based cognitive relay networks has gained great attention, (see, e.g. \cite{Guo2010,Si2011,DuongEL2011,Duong_EL_Dec2012,Blagojevic2012,Li2012}), due to its adaptability to extend the coverage for wireless networks. This situation may arise in practice when data is sent from a cognitive source to a given cognitive destination on a hop-by-hop basis via various intermediate cognitive relay nodes. In
addition to the advantage of extending the coverage without using large power at the transmitter, cognitive relay networks are able to reduce interference causing to PUs.

In particular, Guo \emph{et. al.} in \cite{Guo2010} derived the upper-bound for outage probability of underlay selective DF relay networks operating within the constraint imposed on the peak power received at the primary receiver. In \cite{Si2011}, only based on the partial channel state information (CSI) of primary user link and partial CSI between the relays and the primary user receiver, Li proposed two relay selection schemes for cognitive relay networks. The system performance in terms of outage probability was also provided over Rayleigh fading channels. For amplify-and-forward (AF) relaying, \cite{DuongEL2011} and \cite{Duong_EL_Dec2012} studies the dual hop relaying networks with and without considering combining technique at the secondary destination, respectively. Equipped multi antenna for secondary nodes, the work in \cite{Blagojevic2012} investigated the ergodic capacity for secondary underlay multi-input multi-output (MIMO) networks, where transmit antenna selection (TAS) and maximal ratio combining (MRC)  are used at the transmitter and receiver, respectively. In parallel, Dong Li \cite{Li2012} analyzed the effect of maximum-ratio combining diversity on the performance of underlay single-input multi-output (SIMO) systems where the transmit power and the interference power constraint are taken into accounts. All of the theoretical performance analyses of cognitive relay networks mentioned above have just only been solved for the particular case of two consecutive hops except for the paper \cite{Bao2012}. In \cite{Bao2012}, the authors studied the performance of cognitive underlay multihop decode-and-forward relaying networks over Rayleigh fading channels. However, the results in this paper is limited while it only provided the closed-form expression of the system outage probability under interference constraints. 

Therefore, in this paper, we develop a performance analysis framework for SS-based multi-hop CR networks. In particular, we derive exact closed-form expressions for outage probability (OP) and bit error rate (BER) as well as the approximate closed-form expression for ergodic capacity for the considered system over independent but non-identical distributed (i.n.d.) Rayleigh fading channels. To gain insights, the asymptotic approximation for outage probability, bit error rate at high SNR regime, is also provided. We have shown that the system diversity order is always one regardless the number of hops and the coding gain increases according to the increase of the number of hops. For a pre-determined position of a primary receiver, the problem of relay position optimization is also considered and solved by using the numerical approach. Finally, simulation results are provided to validate the analytical performance assessments.

The remainder of this paper is organized as follows. In Section~\ref{SystemModel}, the system model for underlay cognitive multihop network is described. In Section~\ref{PerformanceAnalysis}, the system performance metrics in terms of outage probability, bit error rate and ergodic capacity are derived for Rayleigh fading channels. In Sect.~\ref{Sec:RelayPositionOptimization}, we are concerned with the problem of relay position optimization. Numerical results are given in Section~\ref{NumericalResults}, where the advantage of cognitive underlay multhop systems is investigated. Finally, conclusions are drawn in Section~\ref{Conclusion}.

\section{System Model}

\begin{figure}[!h]
\centering
\includegraphics[width=\ScaleIfNeeded]{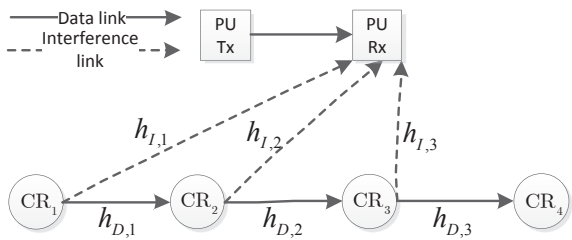}
\caption{System model of a multihop spectrum sharing communication
system.} \label{SystemModel}
\end{figure}

We consider a multi-hop SS system with the coexistence of PUs, i.e.,
licensed users, and SU, i.e., unlicensed users, as shown in
Fig.~\ref{SystemModel}. Assume that PUs and SUs share the same
narrow-band frequency with bandwidth $B$, which is licensed to PUs.
In the primary network, the PU transmitter (PU-Tx) transmits its
data towards the PU receiver (PU-Rx). In the secondary network, the
SU-Tx ($\text{CR}_1$) indirectly transmits the message towards the
SU-Rx ($\text{CR}_{K+1}$) with the help of $K$-1 cognitive
decode-and-forward (regenerative) relays in between denoted by
$\text{CR}_2,\ldots,\text{CR}_K$. Following the SS-based (underlay) approach, SUs are allowed to
concurrently transmit with PU over the same licensed band while
adhering to the interference constraint on the PU-Rx, i.e., any data
transmission of SUs resulting in a higher interference level than an
interference temperature at the PU-Rx is prohibited. To represent
the maximum allowable interference power level at the PU-Rx, the
interference temperature (${\mathcal{I}}_p$), is used \cite{Guo2010,Lee2011,Si2011}. Let $h_{D,k}$ and $h_{I,k}$ be the channel coefficients of the link
from the $k$-th SU-Tx to the next SU-Rx and to the PU-Rx,
respectively. Under Rayleigh fading, $|h_{D,k}|^2$ and $|h_{I,k}|^2$
are exponential distributed with their corresponding parameters
${\lambda _{D,k}} = {\mathbb{E}}\{ |{h_{D,k}}{|^2}\} $ and ${\lambda _{I,k}} =
{\mathbb{E}}\{|{h_{I,k}}{|^2}\}$, where ${\mathbb{E}}\{.\}$ stands for the expectation
operation. Perfect channel state information (CSI) of the
$\text{CR}_k$ $\to$ PU's Rx link is assumed at the
$\text{CR}_k$\footnote{It can be realized by many mechanisms, e.g.,
direct feedback from primary receivers, indirect feedback from band
manager \cite{Peha2005} or using CSI of the $\text{PU-CR}_k$ link
with channel reciprocity property \cite{Zhao2008}.}. We further
assume that the additive white Gaussian noise (AWGN) associated with
each hop is a zero-mean Gaussian random variable with variance
${\mathcal{N}}_0$.

\section{Performance Analysis}
\label{PerformanceAnalysis}

Consider the $k$-th hop of CR multihop networks, the instantaneous
signal-to-noise ratio is written as 
\begin{align}
{\gamma _k} = {P_k}\frac{{{{\left| {{h_{D,k}}} \right|}^2}}}{{{{\cal
N}_0}}}.
\end{align}
To ensure that the interference power at the PU-Rx is always below
the interference temperature, ${\mathcal{I}}_p$, the transmit power is
upper bounded by 
\begin{align}
P_k \leq {{\mathcal{I}}_p}/{|h_{I,k}|^2}.
\end{align}
Aiming to enhance the system performance, we adopt 
\begin{align}
P_k = {{\mathcal{I}}_p}/{|h_{I,k}|^2},
\end{align}
yielding ${\gamma _k} = \frac{{{{\cal I}_p}}}{{{{\cal N}_0}}}\frac{{{{\left|
{{h_{D,k}}}. \right|}^2}}}{{{{\left| {{h_{I,k}}} \right|}^2}}}$.

For Rayleigh fading channels, the probability density function (PDF)
of $\gamma_{{\mathsf{Z}},k}=|h_{{\mathsf{Z}},k}|^2$   with ${\mathsf{Z}}
\in \{D,I\}$ is of the form ${f_{{\gamma _{{\mathsf{Z}},k}}}}(\gamma )
= \frac{1}{{{{\bar \gamma }_{{\mathsf{Z}},k}}}}{e^{ - \frac{\gamma
}{{{{\bar \gamma }_{{\mathsf{Z}},k}}}}}}$, where ${\bar \gamma
}_{{\mathsf{Z}},k}={\lambda _{{\mathsf{Z}},k}}$. The PDF of the received
SNR at hop $k$, $\gamma_k$, is derived as \cite[p. 187, eq. 6-60]{Papoulis2002}
\begin{eqnarray}
\label{PDF_gamma_k} {f_{{\gamma _k}}}(\gamma )
            &=&\int\nolimits_0^{
\infty }\!\!{\frac{x{{\cal N}_0}}{{{\cal I}_p}}} \!{f_{{\gamma
_{D,k}}}}\!\!\left(\!\frac{{x}\gamma{\cal N}_0}{{{\cal
I}_p}}\!\right)\!{f_{{\gamma _{I,k}}}}({x})d{x}\nonumber\\
            &=&\frac{{{\alpha
_k}}}{{{{\left( {\gamma  \!+\! {\alpha _k}} \right)}^2}}},
\end{eqnarray}
where ${\alpha _k} = \frac{{{{\bar \gamma }_{D,k}}}}{{{{\bar \gamma
}_{I,k}}}}{{{{\cal I}_p}} \mathord{\left/
 {\vphantom {{{{\cal I}_p}} {{{\cal N}_0}}}} \right.
 \kern-\nulldelimiterspace} {{{\cal N}_0}}}$.
From \eqref{PDF_gamma_k}, the corresponding cumulative distribution
function (CDF) is given by
\begin{eqnarray}
\label{CDF_gamma_k}
{F_{{\gamma _k}}}(\gamma ) &=& \int\nolimits_0^{\gamma} {{f_{{\gamma _k}}}(x)dx } \nonumber\\
                                               &=& \frac{\gamma}{{\gamma  + {\alpha _k}}}.
\end{eqnarray}
Having the PDF and CDF of each $\gamma_k$ in hands, we are now in a
position to derive the performance metrics of the system including
outage probability, bit error probability, and ergodic capacity.
\subsection{Outage Probability}
In an interference-limited multi-hop regenerative relay system, an
outage event is declared whenever any of the instantaneous forwarded
SNRs of the $K$ hops falls below a given threshold, $\gamma_{th}$. In other words, the overall system outage is dominated by the weakest hop. Thus, the end-to-end (e2e) outage probability is written as \cite{Hasna2003}
\begin{eqnarray}
\label{OP_exact}
{\rm{OP}} = \Pr \left[ {\min ({\gamma _1}, \ldots ,{\gamma _K}) < {\gamma _{th}}} \right].
\end{eqnarray}
Since $\gamma_k$ with $k=1,\ldots,K$ are assumed to be independent of each other, \eqref{OP_exact} is rewriten as
\begin{eqnarray}
\label{OP_exact_1}
{\rm{OP}} =  1 - \prod\limits_{k = 1}^K {\left[ {1 - {F_{{\gamma _k}}}({\gamma _{th}})} \right]}. 
\end{eqnarray}
Substituting \eqref{CDF_gamma_k} into \eqref{OP_exact_1}, we have 
\begin{eqnarray}
\label{OP_exact_2}
{\rm{OP}} = 1 - \prod\limits_{k = 1}^K {\frac{\alpha_k}{\gamma_{th} + \alpha_k}}.
\end{eqnarray}
\begin{theorem}
At high SNR regime, the system OP can be approximated as
\begin{equation}
\label{OP_approx}
{\rm{OP}} \to \frac{{{\gamma _{th}}}}{{{\textstyle{{{{\cal I}_p}} \over {{{\cal N}_0}}}}}}\sum\limits_{k = 1}^K {\frac{{{\lambda _{I,k}}}}{{{\lambda _{D,k}}}}}.
\end{equation}
\end{theorem}
\begin{proof}
We start the proof by using the fact that the cross-terms, ${F_{{\gamma
_k}}}({\gamma _{th}}){F_{{\gamma _l}}}({\gamma _{th}})$ with $k \neq
l$ in \eqref{OP_exact_1}, can be neglected compared to ${F_{{\gamma _k}}}({\gamma _{th}})$
for values of interest. As a result, \eqref{OP_exact_1} can be approximated as 
\begin{eqnarray}
\label{OP_approx_1}
{\rm{OP}} \approx \sum\limits_{k = 1}^K {{F_{{\gamma _k}}}({\gamma _{th}})}.
\end{eqnarray}
Plugging \eqref{CDF_gamma_k} into \eqref{OP_approx_1} and making use $x/{(1+x)} \approx x$ for small $x$, we arrive at the desired result. This completes the proof.  
\end{proof}

From \eqref{OP_approx}, it is worth noting that the outage probability at high SNRs is determined by the channel gain ratios between the data and interference channels rather than the average channel powers. To gain further insights, we prove the following theorem.
\begin{theorem}
\label{theo:2}
The system diversity order and coding gain are $G_\text{d}=1$ and $G_\text{c} = \left({\gamma _{th}}\sum\limits_{k = 1}^K {\frac{{{\lambda _{I,k}}}}{{{\lambda _{D,k}}}}}\right)^{-1}$.
\end{theorem}
\begin{proof}
Observing \eqref{OP_approx}, it is obvious to see that the system diversity gain is one and the coding gain is $\left({\gamma _{th}}\sum\limits_{k = 1}^K {\frac{{{\lambda _{I,k}}}}{{{\lambda _{D,k}}}}}\right)^{-1}$ according to ${\rm OP} \to (G_\text{c} {\mathcal{I}}/{\mathcal{N}}_0)^{-G_\text{d}}$ \cite{Wang2003}.
\end{proof}

From Theorem~\ref{theo:2}, it is worth noting that similar to conventional multihop networks, the system diversity order is always one regardless of number of hops and the increase of hops results in an increase of the system coding gain. 

\subsection{Bit Error Rate}
In this section, we study the most generalized scenario in a multihop
network where all the single-hops in the route have the different
statistical behavior\footnote{Other scenarios including independent and identical distributed (i.i.d.) case are a special case of the network under consideration.}, i.e., all
the links are i.n.d. with different average channel power,
$\bar{\gamma}_p \!\neq\! \bar{\gamma}_q$. Taking into account the fact
that a wrong bit transmission from node $p$ to node $q$ ($q>p$) is
equivalent to an odd number of wrong single-hop bit transmission
between both nodes and employing the recursive error relation, we
have the exact e2e BER of the system as \cite{Morgado2010}
\begin{eqnarray}
\label{BER_e2e} {\overline {{\mathop{\rm BER}\nolimits} } _{e2e}} =
\sum\limits_{p = 1}^K {{{\overline {\text{BER}} }_p}\prod\limits_{q
= p + 1}^K {\left( {1 - 2{{\overline {\text{BER}}}_q}} \right)} },
\end{eqnarray}
where ${\overline {\text{BER}} _p}$ denotes the average BER for
square $M$-ary quadrature amplitude ($M$-QAM) modulation ($M =
{4^m},m = 1,2, \cdots$) in hop $p$ and is given by 
\begin{eqnarray}
\label{BERp}
{\overline {{\mathop{\rm BER}\nolimits} } _p} = \int\limits_0^\infty 
{\rm BER}_{\rm AWGN} f_{\gamma_p}(\gamma) d\gamma.
\end{eqnarray}
In the above equation, ${\rm BER}_{\rm AWGN}$ is the instantaneous BER of hop $p$, namely \cite{Cho2002}
\begin{eqnarray}
\label{BER_agwn}
{\rm BER}_{\rm AWGN} = \frac{{\sum\limits_{j = 1}^{{{\log }_2}\sqrt M } {\sum\limits_{n =
0}^{{\upsilon _j}} {\phi _n^j} } \,{\rm{erfc}}\left( {\sqrt {{\omega
_n}\gamma } } \right)}}{{\sqrt M {{\log }_2}\sqrt M }}, 
\end{eqnarray}
where $\upsilon _j\!=\!(1\!-\!{2^{\!-\!j}})\sqrt M  \!-\! 1$, ${\omega _n} \!=\! \frac{{{{(2n \!+\! 1)}^2}3{{\log }_2}M}}{{2M \!-\! 2}}$, and $\phi_n^j \!=\!
{( - 1)^{\left\lfloor {\frac{{n{2^{j \!-\! 1}}}}{{\sqrt M }}}
\right\rfloor }}\!\!\!\left( {{2^{j \!-\!1}} \!-\!\! \left\lfloor {\frac{{n{2^{j \!-\!
1}}}}{{\sqrt M }} \!+\! \frac{1}{2}} \right\rfloor } \right)$. Here, $\left\lfloor . \right\rfloor$ and ${\rm{erfc}}(x) = \frac{2}{{\sqrt \pi  }}\int\limits_x^\infty  {{e^{ - {t^2}}}dt}$ are defined as the floor and complementary error function, respectively. Substituting \eqref{PDF_gamma_k} and \eqref{BER_agwn} into \eqref{BERp} and swapping integration and summation order, we have\footnote{It should be noted that in this paper we only consider square $M$QAM; however the employed approach could be easily extended for other modulation schemes such as $M$PSK, $M$PAM and rectangular $M$QAM.}
\begin{eqnarray}
\label{BERp_def}
{\overline {{\mathop{\rm BER}\nolimits} } _p}\, = \sum\limits_{j = 1}^{{{\log }_2}\sqrt M } {\sum\limits_{n = 0}^{{\upsilon _j}} {\frac{{\phi _n^j}}{{\sqrt M {{\log }_2}\sqrt M }}{{\cal J}_p}}}, 
\end{eqnarray}
where ${\mathcal{J}}_p$ is defined as follows:
\begin{eqnarray}
{{\cal J}_p} = \int\limits_0^\infty  {{\rm{erfc}}\left( {\sqrt {{\omega _p}\gamma } } \right)\frac{{{\alpha _p}}}{{{{\left( {\gamma  + {\alpha _p}} \right)}^2}}}d\gamma}. 
\end{eqnarray}
Using integration by parts, it is shown that ${\cal J}_p$ is of the form 
\begin{eqnarray}
{{\cal J}_p} = \underbrace {\left. {\frac{{\gamma {\rm{erfc}}\left( {\sqrt {{\omega _p}\gamma } } \right)}}{{\gamma  + {\alpha _p}}}} \right|_{\gamma  = 0}^\infty }_{{{\cal J}_1}} + \frac{{{\omega _p}}}{\pi }\underbrace {\int\limits_0^\infty  {\frac{{\sqrt \gamma  {e^{ - {\omega _p}\gamma }}}}{{\gamma  + {\alpha _p}}}} d\gamma }_{{{\cal J}_2}}.
\end{eqnarray}
To determine ${\cal{J}}_p$, we need to compute ${\cal{J}}_1$ and ${\cal{J}}_2$. For ${\cal{J}}_1$, making use the l'Hopital rule, we have
\begin{eqnarray}
{{\cal J}_1} &=& {\left[ {{\rm{erfc}}\left( {\sqrt {{\omega _p}\gamma } } \right) - \sqrt {\frac{{{\omega _p}\gamma }}{\pi }} {e^{ - {\omega _p}\gamma }}} \right]_{\gamma  = \infty }}\nonumber\\
             &=& 0.
\end{eqnarray}
For ${\mathcal{J}}_2$, introducing a change of variables, namely, $u = \sqrt{\gamma}$, we can rewrite ${\mathcal{J}}_2$ as 
\begin{eqnarray}
{{\cal J}_2} &=& 2\int\limits_0^\infty  {\frac{{{u^2}}}{{{u^2} + b}}} {e^{ - {\omega _p}{u^2}}}du\nonumber\\
             &=& 2\left[ {\int\limits_0^\infty  {{e^{ - {\omega _p}{u^2}}}du}  - \int\limits_0^\infty  {\frac{b}{{{u^2} + b}}} {e^{ - {\omega _p}{u^2}}}du} \right].
\end{eqnarray}
Together with the identities \cite[eq. (3.23.3)]{Gradshteyn2007} and \cite[eq. (7.4.11)]{Abramowitz1972}, it is shown that 
\begin{eqnarray}
\label{J_p}
{{\cal J}_p} = 1 - \sqrt {{\omega _n}{\alpha _k}} {e^{{\omega _n}{\alpha _k}}}\sqrt \pi  {\rm{erfc}}\left( {\sqrt {{\omega _n}{\alpha _k}} } \right).
\end{eqnarray}

Plugging \eqref{J_p} into \eqref{BERp_def}, we achieve the closed-form expression for ${\overline {\text{BER}}_p}$ as
\begin{eqnarray}
\label{BER_k} {\overline {{\mathop{\rm BER}\nolimits} } _p}\!=\!\!
\frac{{\sum\limits_{j = 1}^{{{\log }_2}\sqrt M }\!\!\!{\sum\limits_{n =
0}^{{\upsilon _j}} {\phi _n^j \left[ {1 \!\!-\!\! \sqrt {{\omega _n}{\alpha
_p}} {e^{{\omega _n}{\alpha _p}}}\!\sqrt \pi {\rm{erfc}}\!\left(
{\sqrt {{\omega _n}{\alpha _p}} } \right)} \right]} }}}{{\sqrt M
{{\log }_2}\sqrt M }}.
\end{eqnarray}
Substituting \eqref{BER_k} into \eqref{BER_e2e} yields the average
e2e BER. 

For i.i.d. fading channels, i.e., $\left\{ {{\alpha _p}} \right\}_{p = 1}^K = \alpha$, \eqref{BER_e2e}
simplifies as \eqref{BER_e2e_iid} shown at the top of the next page.
\begin{figure*}[!t]
\begin{eqnarray}
\label{BER_e2e_iid}
\overline{\text{BER}}_{e2e} = \frac{1}{2}\left[ {1 - {{\left( {1 - \frac{2}{{\sqrt M {{\log }_2}\sqrt M }}\sum\limits_{j = 1}^{{{\log }_2}\sqrt M } {\sum\limits_{n = 0}^{{\upsilon _j}} {\phi _n^j\left( {1 - \sqrt {{\omega _n}\alpha } {e^{{\omega _n}\alpha }}\sqrt \pi  {\rm{erfc}}\left( {\sqrt {{\omega _n}\alpha } } \right)} \right)} } } \right)}^K}} \right]
\end{eqnarray}
\setlength{\arraycolsep}{1pt}
\hrulefill \setlength{\arraycolsep}{0.0em}
\vspace*{1pt}
\end{figure*}
\begin{theorem}
At high SNR regime, the end-to-end BER of cognitive underlay multihop DF relaying networks operating in Rayleigh fading channels is approximated as 
\begin{eqnarray}
\label{BER_e2e_approx}
\overline{\rm BER}_{\rm{e2e}} \to \left\{ {\begin{array}{*{20}{c}}
{\frac{a}{{2b}}\sum\limits_{p = 1}^K {\frac{1}{{{\alpha _p}}},} }&{{\rm{i}}{\rm{.n}}{\rm{.d}}{\rm{.}}\,{\rm{channels}}}\\
{\frac{{Ka}}{{2b\alpha }},}&{{\rm{i}}{\rm{.i}}{\rm{.d}}{\rm{.}}\,{\rm{channels}}}
\end{array}} \right.,
\end{eqnarray}
where $a = \frac{{\sqrt M  - 1}}{{\sqrt M {{\log }_2}\sqrt M }}$ and $b = \frac{{3{{\log }_2}M}}{{2(M - 1)}}$.
\end{theorem}
\begin{proof}
Observing \eqref{BER_e2e} and using the fact that at high SNR regime the product term, $\prod\limits_{q = p + 1}^K {\left( {1 - 2{{\overline{\rm BER} }_q}} \right)}$, approaches to one. We are able to approximate the end-to-end BER as 
\begin{eqnarray}
{\overline {{\mathop{\rm BER}\nolimits} } _{e2e}} 
        &=& \sum\limits_{p = 1}^K {{{\overline {\rm BER} }_p}\underbrace {\prod\limits_{q = p + 1}^K {\left( {1 - 2{{\overline {\rm BER} }_q}} \right)} }_{ \to 1}} \nonumber\\
       &\approx& \sum\limits_{p = 1}^K {{{\overline {\rm BER} }_p}}. 
\end{eqnarray}
By neglecting some of the higher order terms in \eqref{BER_agwn},  ${\overline{\rm BER}}_p$ is expressed as \cite[eq. (18)]{Cho2002}
\begin{eqnarray}
{\overline {\rm BER} _p} = \int\limits_0^\infty a\;{\mathop{\rm erfc}\nolimits} \left( {\sqrt {b\gamma } } \right)\frac{{{\alpha _p}}}{{{{(\gamma  + {\alpha _p})}^2}}}d\gamma,
\end{eqnarray}
where $a = \frac{{\sqrt M  - 1}}{{\sqrt M {{\log }_2}\sqrt M }}$ and $b = \frac{{3{{\log }_2}M}}{{2(M - 1)}}$.
Employing the same steps as for \eqref{BER_k} and then making use the infinite series representation for the error function, i.e., ${\mathop{\rm erfc}\nolimits} (\sqrt x ) \approx \frac{{{e^{ - x}}}}{{\sqrt {\pi x} }}\left( {1 - \frac{1}{{2x}}} \right)$ for large $x$ \cite{Tellambura2000}, we can have the desired results as in \eqref{BER_e2e_approx}.

For i.i.d. case, it follows immediately from the result of the i.n.d. case by letting $\alpha  = \left\{ {{\alpha _p}} \right\}_{p = 1}^K$.
\end{proof}
\subsection{Ergodic Capacity}
Besides outage and bit error probability, the ergodic capacity is
another important performance measure, defined as the expected
value of the instantaneous mutual information between the cognitive
source and the cognitive destination. The ergodic capacity $\mathcal{C}$ (in
bits/second) per unit bandwidth can be expressed as
\begin{eqnarray}
\label{Capacity_def}
{\mathcal{C}}
%\!=\! \frac{1}{K}E\left\{ {{{\log }_2}(1
%\!+\! {\gamma _{e2e}})} \right\}
\!=\! \frac{1}{K}\int\nolimits_0^{
\infty} {{{\log }_2}(1 + \gamma ){f_{{\gamma _{e2e}}}}(\gamma
)d\gamma},
\end{eqnarray}
where $\gamma _{e2e}$ is the equivalent instantaneous e2e SNR
for the multi-hop CR network. To derive the ergodic capacity, we
first need an expression for the PDF of $\gamma_{e2e}$. However,
with regenerative relaying, an exact closed-form expression for the PDF is not mathematically viable. For mathematics tractability, we use the approximation approach. According to \cite{Wang2007},
regardless of the modulation scheme used, $\gamma_{e2e}$ can be tightly approximated as 
\begin{eqnarray}
\gamma_{e2e} \approx{\tilde \gamma}_{e2e} = \min_{k=1,\ldots,K}\!\gamma_k.
\end{eqnarray} 
Having been widely adopted in the performance studies of DF relay networks (see, e.g., \cite{Trung2009,Alexandropoulos2010}), the advantage of this analytical approach is able to provide a mathematics tractable form for the CDF and PDF of the end-to-end SNR. Consequently, the PDF of ${\tilde \gamma}_{e2e}$ is given by
\begin{eqnarray}
\label{PDF_gamma_e2e_def} {f_{{\tilde\gamma _{e2e}}}}(\gamma ) 
     &=&\frac{{d{F_{{\tilde\gamma _{e2e}}}}(\gamma )}}{{d\gamma }} \nonumber\\
     &=&\sum\limits_{k = 1}^K {{f_{{\gamma _k}}}(\gamma )\!\!\!\prod\limits_{n=1,
n \ne k}^K \!\!\!{\left( {1 - {F_{{\gamma _n}}}(\gamma )} \right)} },
\end{eqnarray}
where ${F_{{\tilde\gamma _{e2e}}}}(\gamma ) = 1 - \prod\limits_{k =
1}^K {\left[ {1 - {F_{{\gamma _k}}}({\gamma})} \right]}$.
Substituting \eqref{PDF_gamma_k} and \eqref{CDF_gamma_k} into
\eqref{PDF_gamma_e2e_def}, we get
\begin{eqnarray}
\label{PDF_gamma_e2e_1} {f_{{{\tilde \gamma}_{e2e}}}}(\gamma ) =
\sum\limits_{k = 1}^K {\frac{{\prod\nolimits_{n = 1}^K {{\alpha _n}}
}}{{\left( {\gamma  + {\alpha _k}} \right)\prod\nolimits_{n = 1}^K
{\left( {\gamma  + {\alpha _n}} \right)} }}}.
\end{eqnarray}
With the current form of \eqref{PDF_gamma_e2e_1}, it is very difficult to obtain the closed-form expression for the end-to-end capacity. To facilitate the analysis, we sort and renumber
$\alpha_k$ in the ascending order as ${\alpha _1} =  \cdots  =
{\alpha _{{r_1}}} = {\beta _1} < \, \cdots  < {\alpha _{{r_1} +
{r_2} + \cdots  + {r_{N - 1}} + 1}} =  \cdots  = {\alpha _{{r_1} +
{r_2} + \cdots  + {r_N}}} = {\beta _N}$ and  $\sum\nolimits_{n = 1}^N
{{r_n}}  = K$ with $r_k$ being a positive integer. Stated another way, $\beta_1, \cdots, \beta_N$ are distinct elements of $\alpha_1,\cdots,\alpha_K$. Using the
partial-fraction expansion, \eqref{PDF_gamma_e2e_1} can be rewritten as
\begin{eqnarray}
\label{PDF_gamma_e2e}
{f_{{{\tilde \gamma}_{e2e}}}}(\gamma ) = \prod\nolimits_{k = 1}^K
{{\alpha _k}} \left( {\sum\limits_{n = 1}^N {\sum\limits_{l =
1}^{{r_n}} {\frac{{{A_{n,l}}}}{{{{(\gamma  + {\beta _n})}^{l +
1}}}}} } } \right),
\end{eqnarray}
where $A_{n,l}$ is the coefficient of the partial-fraction expansion determined as \cite{Amari1997,Khuong2006}\footnote{For convenience, the coefficients $A_{n,l}$ can be obtained more easily by solving the system of $K$ equations, which is established by randomly choosing $K$ distinct values of $\gamma$ but not equal to any $\beta_n$ \cite{Khuong2006}. Let us denote $K$ chosen values of $\gamma$ as $B_u$ with $u = 1,\cdots,K$, we can obtain a linear system of equations as
\begin{eqnarray}
\sum\limits_{n = 1}^N {\sum\limits_{l = 1}^{{r_n}} {\frac{{{A_{n,l}}}}{{{{(\gamma  + {\beta _n})}^{l + 1}}}}} }  = \sum\limits_{k = 1}^K {\frac{1}{{\left( {\gamma  + {\alpha _k}} \right)\prod\nolimits_{n = 1}^K {\left( {\gamma  + {\alpha _n}} \right)} }}}
\end{eqnarray}
where $\textbf{A} = {[\begin{array}{*{20}{c}}
{A_{1,1}}& \cdots &{\,A_{n,l}}& \cdots &{A_{N,{r_n}}}
\end{array}]^T}$ is obtained by ${\bf{A}} = {\bf{C}}^{-1}{\bf{D}}$, where ${[.]^T}$ denotes the transpose operator; ${\bf{C}}$ is a $K \times K$ matrix, whose entries are ${C_{u,v}} = \frac{1}{{{{({B_u} + {\beta _p})}^{q+1}}}}$ with $v = q + \sum\limits_{m = 1}^{p - 1} {{r_m}}$; ${\bf{D}} = [\begin{array}{*{20}{c}}{{D_1}}& \cdots &{{D_u}}& \cdots &{{D_{K}}{]^T}}
\end{array}$ with ${D_u} = \prod\nolimits_{n = 1}^K {\frac{1}{{\left( {{B_u} + {\alpha _n}} \right)}}} \sum\limits_{k = 1}^K {\frac{1}{{\left( {{B_u} + {\alpha _k}} \right)}}} $ and $u,v=1,\cdots,K$.}
\begin{eqnarray}
{A_{n,l}} = \frac{1}{{({r_n} - l)!}}{\left. {\left\{
{\frac{{{\partial ^{({r_n} - l)}}}}{{\partial {s^{({r_n} -
l)}}}}[{{(\gamma  + {\beta _n})}^{{r_n}}}{f_{{{\tilde
\gamma}_{e2e}}}}(\gamma )]} \right\}} \right|_{\gamma  =  - {\beta
_n}}}.
\end{eqnarray}
Plugging \eqref{PDF_gamma_e2e} in
\eqref{Capacity_def}, we get
\begin{eqnarray}
\label{eq:C_e2e_ind} 
{\cal C} \approx \frac{{\prod\nolimits_{k = 1}^K
{{\alpha _k}} }}{{K}}\sum\limits_{n = 1}^N {\sum\limits_{l =
1}^{{r_n}} A_{n,l} } I_l(\beta_n),\end{eqnarray}
where $I_l(\beta_n)$ with $l \ge 1$ is the auxiliary function defined as 
\begin{align}
{I_l}(\beta_n) = \int\nolimits_0^\infty  {\frac{{{\log_2}(1 + \gamma)}}{{{{(\beta_n + \gamma )}^{l+1}}}}} d\gamma. 
\end{align}
Using integration by parts, we have
\begin{align}
\label{I_n_0}
{I_l}({\beta _n}) =& \underbrace {\left[ { - \frac{{{{\log }_2}(1 + \gamma )}}{{l{{(\gamma  + {\beta _n})}^l}}}} \right]_{\gamma  = 0}^\infty }_{ \to 0}\nonumber\\
                   &+ \frac{1}{{l\ln 2}}\int_0^\infty  {\frac{{d\gamma }}{{(1 + \gamma ){{(\gamma  + {\beta _n})}^l}}}}.
\end{align}
When $l$ is an integer, after using partial fraction expansion, we have \eqref{I_l_1} as shown at the top of the next page.
\begin{figure*}[!t]
\begin{align}
\label{I_l_1}
{I_l}({\beta _n}) = \frac{1}{{l\ln 2}}\int\limits_0^\infty  {\left[ {\frac{1}{{{{({\beta _n} - 1)}^l}(\gamma  + 1)}} - \sum\limits_{k = 1}^l {\frac{1}{{{{({\beta _n} - 1)}^k}{{(\gamma  + {\beta _n})}^{l + 1 - k}}}}} } \right]} d\gamma 
\end{align}
\setlength{\arraycolsep}{1pt}
\hrulefill \setlength{\arraycolsep}{0.0em}
\vspace*{1pt}
\end{figure*}
Note that the integral in \eqref{I_l_1} is not converging due to the first term. To deal with the problem, by appropriate rearrangements and then performing the integrations, we obtain the final closed-form expression for $I_l(\beta_n)$ as \eqref{I_l_2}.

\begin{figure*}[!t]
\begin{align}
\label{I_l_2} 
{I_l}({\beta _n}) &= \frac{1}{{l\ln 2}}\left[ {\frac{1}{{{{({\beta _n} - 1)}^l}}}\int\limits_0^\infty  {\left( {\frac{1}{{\gamma  + 1}} - \frac{1}{{\gamma  + {\beta _n}}}} \right)} d\gamma  - \sum\limits_{k = 1}^{l - 1} {\frac{1}{{{{({\beta _n} - 1)}^k}}}\int\limits_0^\infty  {\frac{{d\gamma }}{{{{(\gamma  + {\beta _n})}^{l + 1 - k}}}}} } } \right]\nonumber\\
                  &= \frac{1}{{l\ln 2}}\left[ {\frac{{\log {\beta _n}}}{{{{({\beta _n} - 1)}^l}}} - \sum\limits_{k = 1}^{l - 1} {\frac{1}{{{{({\beta _n} - 1)}^k}(l - k)\beta _n^{l - k}}}} } \right]
\end{align}
\setlength{\arraycolsep}{1pt}
\hrulefill \setlength{\arraycolsep}{0.0em}
\vspace*{1pt}
\end{figure*} 
For the special case of $\beta_n=1$, \eqref{I_l_2} simplifies to
\begin{eqnarray}
\label{I_l_3}
{I_l}(\beta_n) &=& \frac{1}{{l\ln 2}}\int\nolimits_0^\infty  {\frac{{d\gamma
}}{{{{(\gamma  + 1)}^{l + 1}}}}},\nonumber\\
               &=& \frac{1}{{{l^2}\ln 2}}.
\end{eqnarray}
For i.i.d. case, from \eqref{PDF_gamma_e2e_def}, we have
\begin{align}
\label{eq:PDF_gamma_e2e_iid}
{f_{{{\tilde \gamma }_{e2e}}}}(\gamma ) &= K{\left[ {1 - {F_{{\gamma _k}}}(\gamma )} \right]^{K - 1}}{f_{{\gamma _k}}}(\gamma )\nonumber\\
                                        &= \frac{{K{\alpha ^K}}}{{{{(\gamma  + \alpha )}^{K + 1}}}}.
\end{align}
Combining \eqref{eq:PDF_gamma_e2e_iid} and \eqref{Capacity_def}, we have the e2e ergodic capacity for this case as
\begin{align}
\label{eq:C_e2e_iid}
{\cal C} &= {\alpha ^K}\int\limits_0^\infty  {\frac{{{{\log }_2}(1 + \gamma )}}{{{{(\gamma  + \alpha )}^{K + 1}}}}} \nonumber\\
         &= {\alpha ^K}{I_K}(\alpha ).
\end{align}
Plugging \eqref{I_l_2} (or \eqref{I_l_3}) into \eqref{eq:C_e2e_ind} (or \eqref{eq:C_e2e_iid}), we have the closed-form, integral-free, expression for the system capacity. It is worth noting that our suggested method is precise and tractable with the determination of the appropriate parameters being done straightforwardly. Additionally, as $\mathcal{C}$ is given in a closed-form fashion, its evaluation is instantaneous regardless of the number of hops and the value of the maximum interference temperature.

It is of interest to compare the ergodic capacity of underlay DF and AF multihop networks. The following theorem is provided to answer such the question. 
\begin{theorem}
For the same network and channel settings of underlay multihop networks, DF relaying provides slightly better ergodic capacity than its AF counterpart. 
\end{theorem}
\begin{proof}
Denote ${{\cal C}_{{\rm{DF}}}}$ and ${{\cal C}_{{\rm{AF}}}}$ be the ergodic capacity for the underlay DF and AF multihop network, respectively. According to the min-cut max-flow theorem \cite{Cover1979}, namely the end-to-end system capacity cannot be larger than the capacity of each hop, the end-to-end ergodic capacity for underlay DF multihop networks can be written as
\begin{eqnarray}
{\cal C}_\text{DF} \le \min ({c_1},{c_2}, \ldots ,{c_K}),
\end{eqnarray}
where $c_k$ with $k=1,\ldots,K$ is the Shannon capacity of hop $k$, given by 
\begin{eqnarray}
\label{C_k}
{c_k} = \frac{1}{K}\int\limits_0^\infty  {{{\log }_2}(1 + {\gamma _k})} {f_{{\gamma _k}}}(\gamma )d\gamma. 
\end{eqnarray}
Making use the Jensen's inequality, we easily see that 
\begin{align}
{{\cal C}_{{\rm{DF}}}} > \frac{1}{K}{{\mathbb{E}}_{{\gamma _1}, \ldots ,{\gamma _K}}}\left\{ {\min \left[ {{{\log }_2}(1 + {\gamma _1}), \ldots ,{{\log }_2}(1 + {\gamma _K})} \right]} \right\}.
\end{align}
Since the binary logarithm is strictly concave, ${{\cal C}_{{\rm{DF}}}}$ can be rewritten as
\begin{align}
\label{C_DF}
{{\cal C}_{{\rm{DF}}}} > \frac{1}{K}{{\mathbb{E}}_{{\gamma _1}, \ldots ,{\gamma _K}}}\left\{ {{{\log }_2}\left[ {1 + \min ({\gamma _1}, \ldots ,{\gamma _K})} \right]} \right\}.
\end{align}
Based on the results reported in \cite{Hasna2004}, i.e.  
\begin{align}
\label{C_AF}
{{\cal C}_{{\rm{AF}}}} < \frac{1}{K}{{\mathbb{E}}_{{\gamma _1}, \ldots ,{\gamma _K}}}\left\{ {{{\log }_2}\left[ {1 + \min ({\gamma _1}, \ldots ,{\gamma _K})} \right]} \right\},
\end{align}
we can complete the proof after making a comparison between \eqref{C_DF} and \eqref{C_AF}.
\end{proof}

Before moving on to the next section, here we would like stress that although DF performs slightly better ergodic capacity, it needs more complicated implementation as compared to AF \cite{Hasna2003}. 

\section{Relay Position Optimization}
\label{Sec:RelayPositionOptimization}

In this section, we focus on the problem of relay position optimization. In particular, for given underlay DF multihop network parameters including the coordinates of the primary receiver, the secondary source, secondary destination, and the number of hops, our problem is to find optimal positions for relays, which makes the system performance (in terms of outage probability or system error probability) minimize. We first provide the solution for the case of outage probability.

\begin{figure}[!h]
\centering
\includegraphics[width=\ScaleIfNeeded]{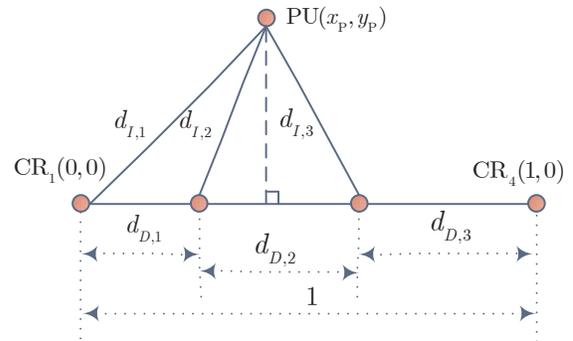}
\caption{Cognitive underlay 3-hop DF relay network in a straight line.} 
\label{SystemModel_opt}
\end{figure}

For simplicity, we consider the network scenario illustrated in Fig.~\ref{SystemModel_opt}, where all secondary nodes are ordered in the sequence and positioned along the straight line connecting the secondary source and the secondary destination. Such a model is mathematically tractable and well-adopted in the literature in studying multihop networks \cite{Oyman2007, Sikora2006}. In addition, it is readily extended to the more generalized case of two dimension (2-D) networks. Interestingly enough, although the linear multihop network model is slightly simplified model of the real world, we can find it in practical, e.g. the communication between cars on a highway, or the communication between road-side units placed along the road \cite{Panichpapiboon2008,Wu2009}.

We further assume that the overall distance between the source and the destination is normalized to one, i.e., 
\begin{align}
\label{eq:nor_dis}
{d_{D,1}} +  \cdots  + {d_{D,K}} = 1,
\end{align}
where $d_{D,k}$ denotes the physical distance of data hop $k$. Under a predetermined position of the primary receiver and a fixed number of hops $K$, the problem of finding the optimal position of the relays that minimizes the system outage probability can be mathematically stated as follows:
\begin{eqnarray}
\label{eq:Opt_Pro_Def}
\begin{array}{l}
\min \,\,\frac{{{\gamma _{th}}}}{{{\textstyle{{{{\cal I}_p}} \over {{{\cal N}_0}}}}}}\sum\limits_{k = 1}^K {\frac{{{\lambda _{I,k}}}}{{{\lambda _{D,k}}}}} \\
{\text{subject to}}\,\left\{ {\begin{array}{*{20}{c}}
{\sum\limits_{k = 1}^K {{d_{D,k}}}  = 1}\\
{{d_{D,k}} > 0,\,\,k = 1, \ldots ,K}
\end{array}} \right..
\end{array}
\end{eqnarray}

Based on a single-slope distance-dependent path loss model for the average channel powers \cite{Stuber2001}, we can write
\begin{align}
\label{eq:dI_dD}
\frac{{{\lambda _{I,k}}}}{{{\lambda _{D,k}}}} = \frac{{{d_{I,k}}^{ - \eta }}}{{{d_{D,k}}^{ - \eta }}}\,\, = {\left( {\frac{{{d_{D,k}}}}{{{d_{I,k}}}}} \right)^\eta },
\end{align}
where $\eta \ge 2$ denotes the path loss exponent. Recalling that $\eta$ typically has value of 2 in free-space environments and up to 5 and 6 in shadowed areas and obstructed in-building scenarios, respectively \cite[Table 4.2]{Rappaport2002}. Combining \eqref{eq:Opt_Pro_Def} and \eqref{eq:dI_dD} and noting that $\gamma_{th}$  and ${\textstyle{{{{\cal I}_p}} \over {{{\cal N}_0}}}}$ are the given constraints, the optimization problem in \eqref{eq:Opt_Pro_Def} can be simplified as 
\begin{align}
\label{eq:Opt_Pro_1}
&\min \,\sum\limits_{k = 1}^K {{{\left( {\frac{{{d_{D,k}}}}{{{d_{I,k}}}}} \right)}^\eta }} \nonumber\\
&{\rm{subject}}\,{\rm{to}}\left\{ {\begin{array}{*{20}{l}}
{\sum\limits_{k = 1}^K {{d_{D,k}}}  = 1,}\\
{{d_{D,k}} > 0,\,\,k = 1, \ldots ,K}
\end{array}} \right..
\end{align}
\begin{theorem}
\label{theo:5}
For a given coordinate of primary receiver $({x_{\rm{P}}},{y_{\rm{P}}})$, the secondary network outage probability achieves its minimum at 
\begin{align}
{\rm{O}}{{\rm{P}}_{{\rm{min}}}} = \frac{{{\gamma _{th}}}}{{{\textstyle{{{{\cal I}_p}} \over {{{\cal N}_0}}}}}}K{\left( {\prod\limits_{k = 1}^K {\frac{{d_{D,k}^*}}{{d_{I,k}^*}}} } \right)^{\frac{\eta }{K}}},
\end{align}
where $d_{D,k}^*$ with $k = 1, \ldots ,K$ are the optimal distance of hop $k$, being the roots of the nonlinear system of $K$ equations as follows:
\begin{align}
\left\{ {\begin{array}{*{20}{c}}
{{d_{D,1}} +  \cdots  + {d_{D,K}} - 1 = 0}\\
{{{({x_{\rm{P}}} - {d_{D,1}})}^2} - \left( {{x_{\rm{P}}}^2 + {y_{\rm{P}}}^2} \right){{\left( {\frac{{{d_{D,2}}}}{{{d_{D,1}}}}} \right)}^2} + {y_{\rm{P}}}^2 = 0}\\
 \vdots \\
{{{\left({x_{\rm{P}}} - \sum\limits_{k = 1}^{K - 1} {{d_{D,k}}} \right)}^2} - \left( {{x_{\rm{P}}}^2 + {y_{\rm{P}}}^2} \right){{\left( {\frac{{{d_{D,K}}}}{{{d_{D,1}}}}} \right)}^2} + {y_{\rm{P}}}^2 = 0}
\end{array}} \right..
\end{align}
And $d_{I,k}^*$ with $k = 2, \ldots ,K$ are determined by using the relationship
\begin{align}
\frac{{d_{D,1}^*}}{{d_{I,1}^*}} =  \cdots  = \frac{{d_{D,k}^*}}{{d_{I,k}^*}} =  \cdots  = \frac{{d_{D,K}^*}}{{d_{I,K}^*}}.
\end{align}
\end{theorem}
\begin{proof}
To solve the above optimization problem, we can use Cauchy theorem\footnote{In some mathematics books, this theorem also is named as the inequality of arithmetic and geometric means.}. In particular, from \eqref{eq:Opt_Pro_1}, we can have 
\begin{align}
\sum\limits_{k = 1}^K {{{\left( {\frac{{{d_{D,k}}}}{{{d_{I,k}}}}} \right)}^\eta }}  \ge K\sqrt[K]{{\prod\limits_{k = 1}^K {{{\left( {\frac{{{d_{D,k}}}}{{{d_{I,k}}}}} \right)}^\eta }} }}.
\end{align}
Equality holds if and only if 
\begin{align}
\label{eq:equ_cond}
\frac{{d_{D,1}^{}}}{{d_{I,1}^{}}} =  \cdots  = \frac{{d_{D,k}^{}}}{{d_{I,k}^{}}} =  \cdots  = \frac{{d_{D,K}^{}}}{{d_{I,K}^{}}}.
\end{align}
Denoting $d_{D,k}^*$  and $d_{I,k}^*$ as the optimal values making the equality occur, the minimized OP is 
\begin{align}
{\rm{O}}{{\rm{P}}_{{\rm{min}}}} = \frac{{{\gamma _{th}}}}{{{\textstyle{{{{\cal I}_p}} \over {{{\cal N}_0}}}}}}K{\left( {\prod\limits_{k = 1}^K {\frac{{d_{D,k}^*}}{{d_{I,k}^*}}} } \right)^{\frac{\eta }{K}}}
\end{align}

To determine  $d_{D,k}^*$  and $d_{I,k}^*$, observing Fig.~\ref{SystemModel_opt} and making use the Pythagorean theorem, we have
\begin{align}
\label{eq:Pytago}
{d_{I,k}}^2 = {y_{\rm{P}}}^2 + {\left( {{x_{\rm{P}}} - \sum\nolimits_{k = 1}^{K - 1} {{d_{D,k}}} } \right)^2}
\end{align}
for $k=2,\ldots,K$. It is obvious for the case of $k=1$ that ${d_{I,1}} = \sqrt {{x_{\rm{P}}}^2 + {y_{\rm{P}}}^2}$. Using \eqref{eq:equ_cond}, \eqref{eq:Pytago} is rewritten as follows: 
\begin{align}
\label{eq:Pytago:2}
{\left( {{x_{\rm{P}}} - \sum\limits_{p = 1}^{k - 1} {{d_{D,k}}} } \right)^2} + {y_{\rm{P}}}^2 = \left( {{x_{\rm{P}}}^2 + {y_{\rm{P}}}^2} \right){\left( {\frac{{{d_{D,k}}}}{{{d_{D,1}}}}} \right)^2}.
\end{align}
Combining \eqref{eq:nor_dis} and \eqref{eq:Pytago:2}, a system of $K$ equations for $d_{D,1},\ldots,d_{D,K}$ is formulated as follows:
\begin{align}
\label{eq:Sys_o_Eq}
\left\{ {\begin{array}{*{20}{c}}
{{d_{D,1}} +  \cdots  + {d_{D,K}} - 1 = 0}\\
{{y_{\rm{P}}}^2 + {{({x_{\rm{P}}} - {d_{D,1}})}^2} - \left( {{x_{\rm{P}}}^2 + {y_{\rm{P}}}^2} \right){{\left( {\frac{{{d_{D,2}}}}{{{d_{D,1}}}}} \right)}^2} = 0}\\
 \vdots \\
{{y_{\rm{P}}}^2 + {{\left({x_{\rm{P}}} - \sum\limits_{k=1}^{K-1}{d_{D,k}}\right)}^2} - \left( {{x_{\rm{P}}}^2 + {y_{\rm{P}}}^2} \right){{\left( {\frac{{{d_{D,K}}}}{{{d_{D,1}}}}} \right)}^2} = 0}
\end{array}} \right..
\end{align}
With the current form of \eqref{eq:Sys_o_Eq}, it seems impossible to obtain the closed-form expression for $d_{D,k}$. Consequently, in this case, the only possibility is to solve \eqref{eq:Sys_o_Eq} numerically.  Using Newton's method \cite{Chong2004}, \eqref{eq:Sys_o_Eq} can be determined by means of the recursion. The initial value for $d_{D,k}$ can be selected as $d_{D,k}^*(0) = \frac{1}{K}$ for all $k$. From \eqref{eq:equ_cond} and \eqref{eq:Sys_o_Eq}, we can complete the proof. 
\end{proof}
We next present optimal relay positions, which minimize the average system bit error rate. Different from outage probability, which serves as a lower bound to the frame error rate for block fading environment and provides an insight into the theoretic-information performance limit, the average bit error rate shows the actual system performance for a desired target spectral efficiency, i.e. modulation level. As such, the following theorem is of importance in this regard.
\begin{theorem}
\label{theo:6}
For a predetermined coordinate of primary receiver $({x_{\rm{P}}},{y_{\rm{P}}})$, linear DF multihop networks under interference constraints provide the best performance in terms of bit error probability if and only if 
\begin{align}
\frac{{d_{D,1}^{}}}{{d_{I,1}^{}}} =  \cdots  = \frac{{d_{D,k}^{}}}{{d_{I,k}^{}}} =  \cdots  = \frac{{d_{D,K}^{}}}{{d_{I,K}^{}}}.
\end{align}
And the corresponding system bit error probability under optimal relay positions is 
\begin{align}
\overline{\rm BER}_{e2e} = \frac{a}{2b}K{\left( {\prod\limits_{k = 1}^K {\frac{{d_{D,k}^*}}{{d_{I,k}^*}}} } \right)^{\frac{\eta }{K}}}.
\end{align}
\end{theorem}
\begin{proof}
The proof is omitted here due to the similarity of the form between OP and BER at the high SNR regime. Then the Theorem~\ref{theo:6} is easily inferred from Theorem~\ref{theo:5}
\end{proof}

From Theorem~\ref{theo:5} and~\ref{theo:6}, it is worthy to point out that the optimal relay positions (both in terms of outage probability and bit error rate) do not depend on the path loss exponent . 

\section{Numerical Results and Discussion}
\label{NumericalResults}
The purpose of this section is twofolds. We first provide numerical results to confirm the derived analytical expressions and then show the network performance advantage offered by relay position optimization. 

For illustrative purpose, we consider a linear multi-hop network in
a 2-D plane, where all SUs are co-linearly located and the distance
between the cognitive source and the cognitive destination is
normalized to one. Furthermore, the cognitive source and the
cognitive destination are located at points with coordinates (0,0)
and (1,0), respectively.  Each cognitive relay node is equidistant
from each other, i.e. ${d_{\text{CR}_k,\text{CR}_{k + 1}}} = 1/K$.
The average channel power for the transmission between node
$\mathsf{A}$ and node $\mathsf{B}$ is modeled as ${\lambda
_{\mathsf{A},\mathsf{B}}} = {d_{\mathsf{A},\mathsf{B}}}^{ - \eta }$
where $\eta$ denotes the path loss exponent with $\mathsf{A} \in
\left\{{\text{CR}_1, \ldots ,\text{CR}_{K - 1}} \right\}$ and
$\mathsf{B} \in \left\{ {\text{PU},\text{CR}_2,\ldots,\text{CR}_K}
\right\}$. In all examples, we locate the PU-Rx at coordinate (0.35,
0.35) and set $\eta = 4$.

\begin{figure}[!h]
\centering
\includegraphics[width=\ScaleIfNeeded]{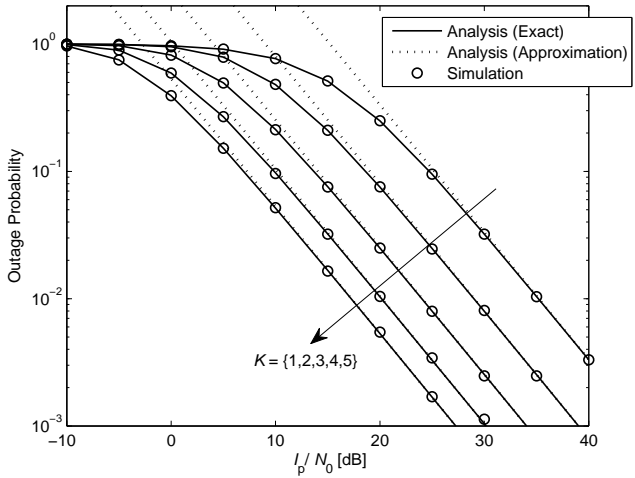}
\caption{Outage probability versus average ${\mathcal{I}}_p/{\mathcal{N}}_0$.} 
\label{Fig1}
\end{figure}
\begin{figure}[!h]
\centering
\includegraphics[width=\ScaleIfNeeded]{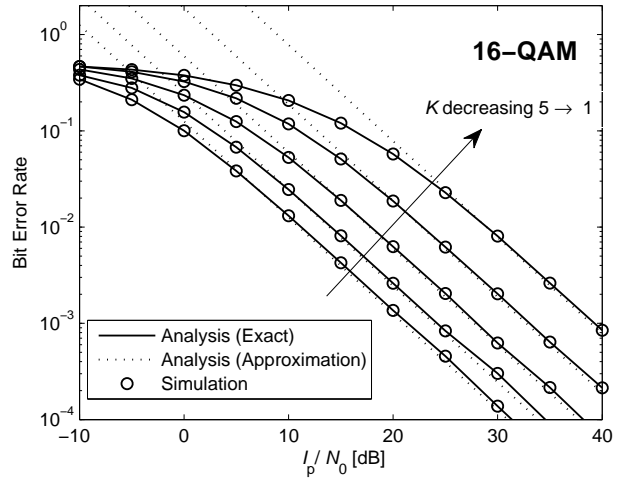}
\caption{Bit error probability versus average ${\mathcal{I}}_p/{\mathcal{N}}_0$.} 
\label{Fig2}
\end{figure}

In Fig.~\ref{Fig1} and Fig.~\ref{Fig2}, we respectively illustrate
the outage probability and BER of the multi-hop cognitive networks
as a function of interference temperature for different number of
hops. As can be observed from the two figures, the performance is
enhanced as the number of cognitive relay hops $K$ increases. It is
important to note that the diminishing gain returns as the number of
hops increases. The numerically evaluated results demonstrate the
correctness of the presented analysis.

\begin{figure}[!h]
\centering
\includegraphics[width=\ScaleIfNeeded]{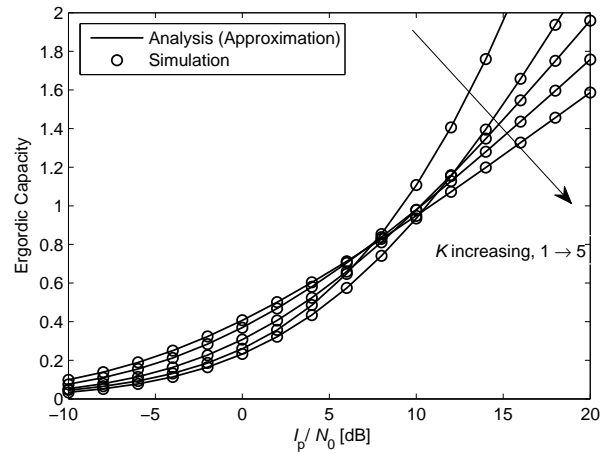}
\caption{Ergodic capacity versus average
${\mathcal{I}}_p/{\mathcal{N}}_0$.} \label{Fig3}
\end{figure}

Figure~\ref{Fig3} displays the capacity performance for cognitive
multi-hop transmission by varying the number of hops, $K=1,2,\ldots,
5$. It is worth noting that for interference-limited regime (low
${\mathcal{I}}_p/{\mathcal{N}}_0$), the system with large $K$ offers
improved performance. For a high interference temperature level,
multi-hop transmission with small hops is more favorable. It can be
explained by using the fact that with the channel model and
time-sharing schedule used, at low ${\mathcal{I}}_p/{\mathcal{N}}_0$
transmission over shorter distance corresponds to increased
effective SNRs while at high ${\mathcal{I}}_p/{\mathcal{N}}_0$
increasing the number of hops is equivalent to reducing the
effective transmission bandwidth of each hop. In addition, the numerical results 
show that the analytical results are in good agreement with the simulation results. 

% Fig. 4
\begin{figure}[!h]
\centering
\includegraphics[width=\ScaleIfNeeded]{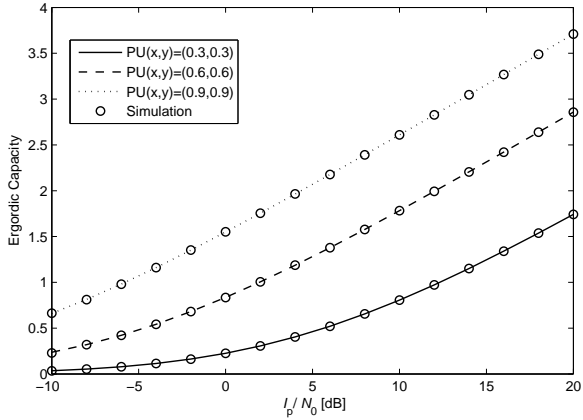}
\caption{Effect of PU-Rx's location on ergodic capacity.}
\label{Fig4}
\end{figure}

Fig.~\ref{Fig4} compares the capacity performance of multi-hop cognitive relay networks for different positions of the PU-Rx given the same number of hops. Observing the results in the figure, we can see that the system performance improves when the primary node is located farther away from the secondary relay
transmitters, as expected.

Up to this point, we have not studied the effect of the proposed relay position optimization. In doing so, we considerer three relay position profiles: randomization, equalization and optimization, denoted as Profile A, Profile B and Profile C, respectively. In Profile A, all secondary relays are chosen randomly from a uniform distribution. In Profile B, the distance between any two nodes is the same. And in Profile C, secondary relays are set using the rule, proposed in Sect.~\ref{Sec:RelayPositionOptimization}. Table~\ref{tab:Table1} demonstrates the results for $K=2$ and 4. 
\begin{table}[htbp]
\centering
\begin{tabular}{|l||c|c|c|}
\hline
    & Profile A & Profile B & Profile C \\ 
\hline
\hline
\multirow{2}{*}{$K=2$} 
 & $d_{D,1}=0.1767$ & $d_{D,1}=0.5$ & $d_{D,1}=0.4192$  \\
 & $d_{D,2}=0.8333$ & $d_{D,2}=0.5$ & $d_{D,2}=0.5808$  \\
\hline
\multirow{4}{*}{$K=4$} 
 & $d_{D,1}=0.20$ & $d_{D,1}=0.25$ & $d_{D,1}=0.1915$  \\
 & $d_{D,2}=0.28$ & $d_{D,2}=0.25$ & $d_{D,2}=0.1900$  \\
 & $d_{D,3}=0.36$ & $d_{D,3}=0.25$ & $d_{D,3}=0.2492$  \\
 & $d_{D,4}=0.16$ & $d_{D,4}=0.25$ & $d_{D,4}=0.3693$  \\
\hline 
\end{tabular}
\caption{Comparision of three relay position profiles}
\label{tab:Table1}
\end{table}
%
% Fig. 5
\begin{figure}[!h]
\centering
\includegraphics[width=\ScaleIfNeeded]{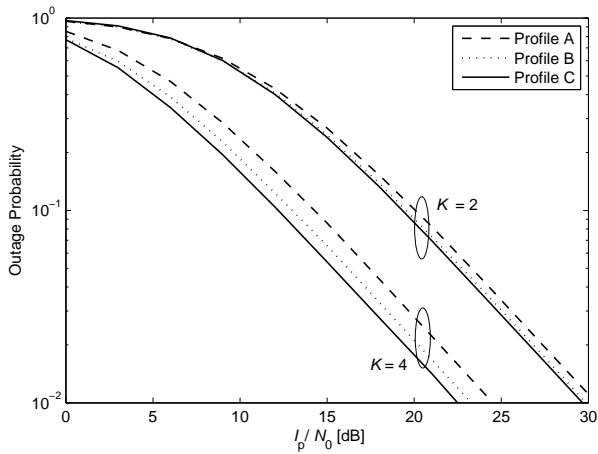}
\caption{Comparison of three relay position profiles.}
\label{Fig5}
\end{figure}

% Fig. 7
\begin{figure}[!h]
\centering
\includegraphics[width=\ScaleIfNeeded]{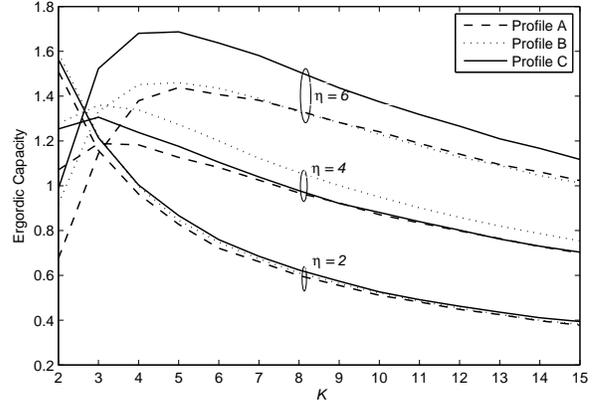}
\caption{Effect of path loss exponent on ergodic capacity.}
\label{Fig7}
\end{figure}

In Fig.~\ref{Fig5}, we can see that Profile C outperforms Profile B, which, in turns, outperforms Profile A.  This observation also repeats in Fig.~\ref{Fig7}, where the effect of path loss exponent is investigated. From Fig.~\ref{Fig7},  we can see that the optimal number of hops is a complicated function of $\eta$. Furthermore, the advantage of profile C, i.e. the capacity gap, becomes bigger with higher value of $\eta$. With small $\eta$, the increase of hops results in the loss of the ergodic capacity. However, with large values of $\eta$, there exists a value of $K$, that makes the system ergodic capacity maximize. It can be explained by using the fact that  with small $\eta$, the benefit of path loss gain is not enough to compensate the loss due to the use of multi orthogonal time slots for multihop communications. 

\section{Conclusion}
\label{Conclusion}
We have investigated the performance of cognitive regenerative multi-hop relay networks using the underlay approach. We have derived the closed-form expressions for the outage probability, BER, and ergodic capacity over i.n.d. Rayleigh fading channels. High analysis for outage probability and bit error rate have also made to provide insights into the system behaviors. The numerical results show that under the interference constraints inflicted by the primary network, the multi-hop transmission still offers a considerable gain as compared to direct transmission and thus makes it an attractive proposition for cognitive networks.
\section*{\uppercase{Acknowledgments}}
This research was supported by the Vietnam National Foundation for Science and Technology Development
(NAFOSTED) (No. 102.01-2011.22).

\bibliographystyle{IEEEtran}
\bibliography{IEEEabrv,reference}

\epsfysize=3.2cm
\begin{biography}{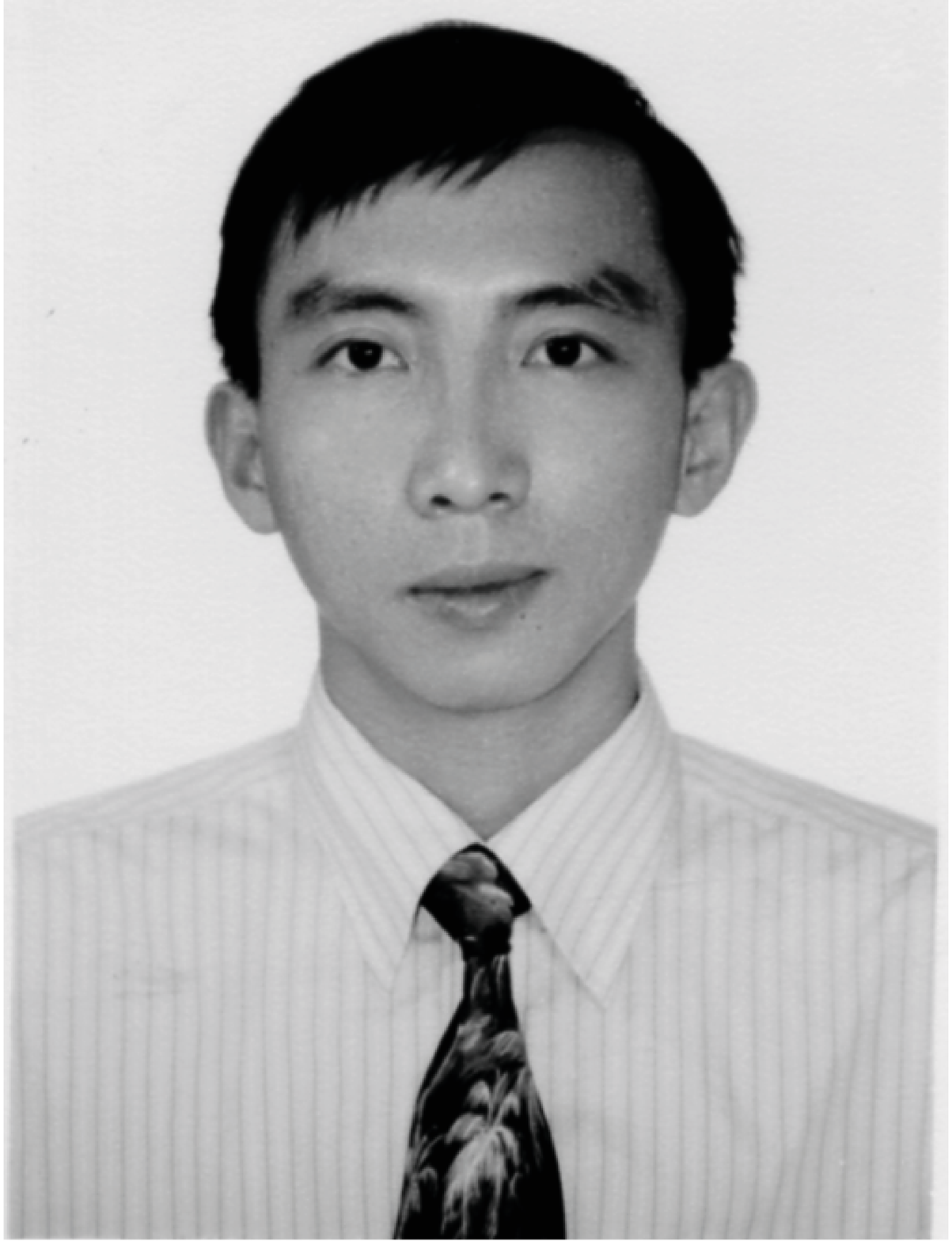}{Vo Nguyen Quoc Bao} was born in Nha Trang, Khanh  Hoa Province, Vietnam, in 1979. He received the B.E. and M.Eng. degree in electrical engineering from Ho Chi Minh City University of Technology (HCMUT), Vietnam, in 2002 and 2005, respectively, and Ph.D. degree in electrical engineering from University of Ulsan, South Korea, in 2010. In 2002, he joined the Department of Electrical Engineering, Posts and Telecommunications Institute of Technology (PTIT), as a lecturer. Since February 2010, he has been with the Department of Telecommunications,  PTIT, where he is currently an Assistant Professor. He is the author or coauthor of more than 60 technical papers in the area of wireless and mobile communications. His major research interests are modulation and coding techniques, MIMO systems, combining techniques, cooperative communications, and cognitive radio. Dr. Bao is a member of Korea Information and Communications Society (KICS), The Institute of Electronics, Information and Communication Engineers (IEICE) and the Institute of Electrical and Electronics Engineers (IEEE).
\end{biography}

\begin{biography}{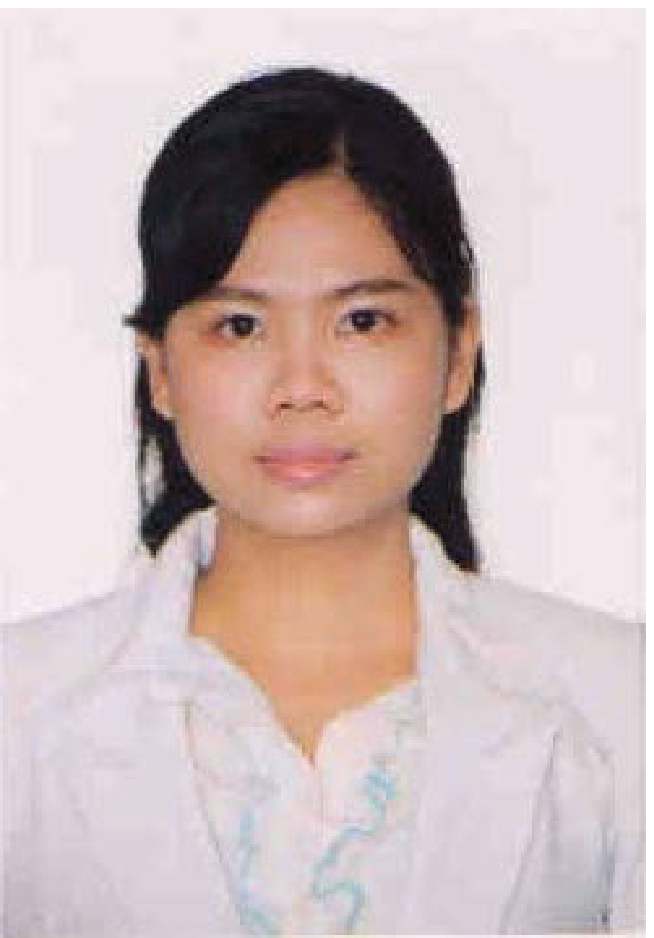}{Tran Thien Thanh} received the B.E. and M.Eng. degree in Electronics and Telecommunications Engineering from Ho Chi Minh City University of Technology, Vietnam, in 2003 and 2011, respectively. Currently she is a lecturer at Faculty of Electrical \& Electronic Engineering, Ho Chi Minh City University of Transport (UT-HCMC). Since Sep. 2011, she has been working toward the PhD degree in Faculty of Electrical and Electronics Engineering, the Ho Chi Minh city University of Technology. Her research interests include the areas of communication theory, signal processing, and networking.
\end{biography}

\begin{biography}{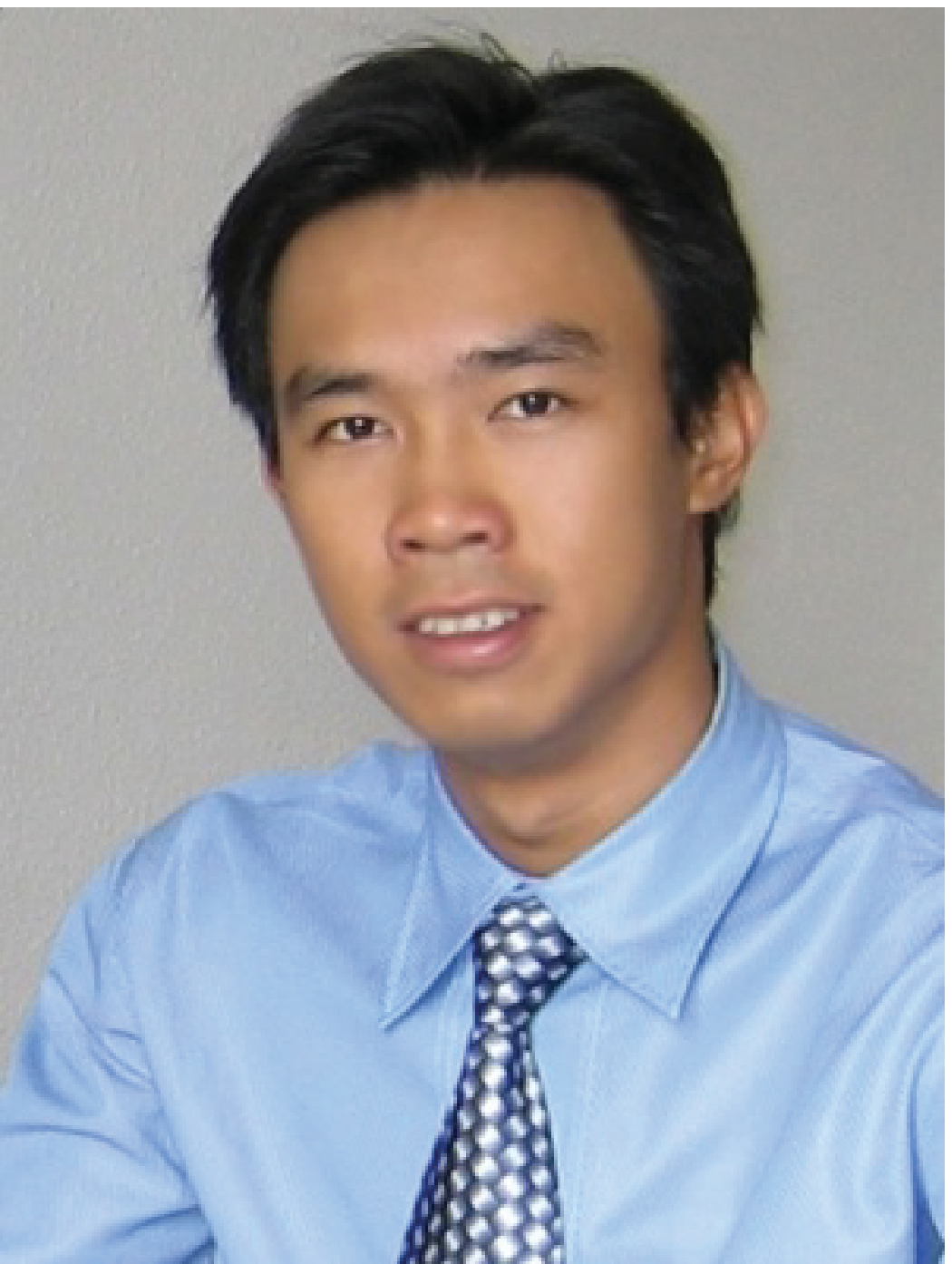}{Nguyen Tuan Duc} received the M.Sc. degree in electrical engineering from Telecom ParisTech University, France, and the Ph.D. degree in signal processing for telecommunication from the University of Rennes 1, France, in 2005 and 2009, respectively. In 2009, he was a Postdoctoral Researcher in cooperative communications for wireless sensor networks with the Institut de Recherche en Informatique et Systemes Aleatoires (IRISA) Research Center, France. Since 2010, he has been a Lecturer and Researcher with the School of Electrical Engineering, Ho Chi Minh City International University, Vietnam National University, Ho Chi Minh City, Vietnam. His research interests include cooperative communications, wireless sensor networks, energy constrained wireless networks, and wireless ad hoc networks.
\end{biography}

\begin{biography}{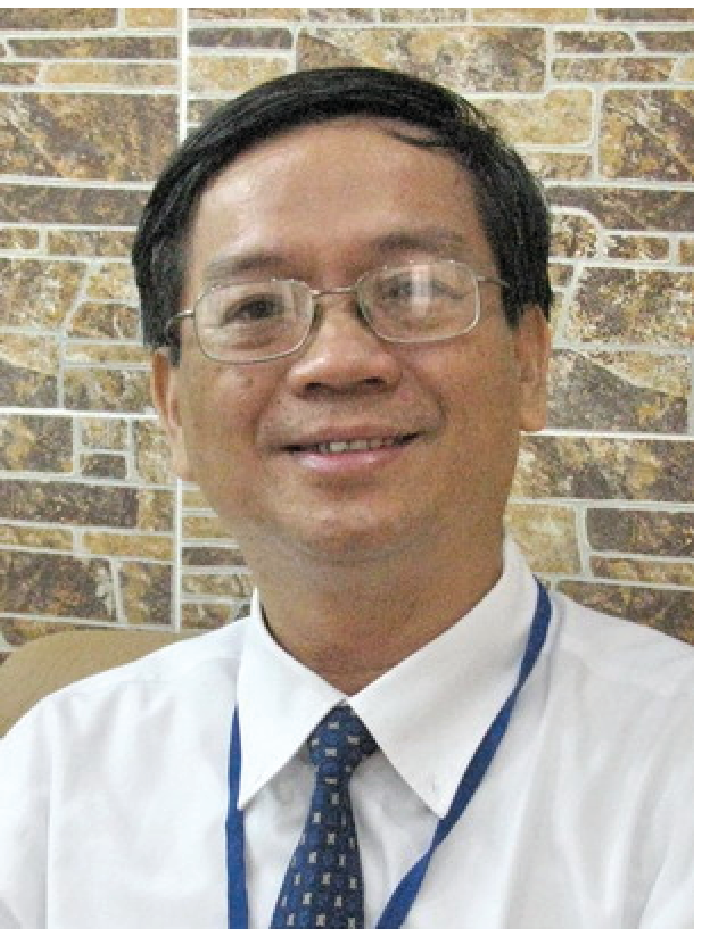}{Thanh Dinh Vu} received the B.E. in HoChiMinh City University of Technology, Vietnam, in 1982 and the M.Sc. and PhD. Degrees in the National Polytechnical Institute of Grenoble, France in 1989 and 1993, respectively. Since then, he has been Professor of the HoChiMinh City University of Technology, Vietnam. His teaching and research fields lie on the microwave techniques and technologies, communication systems and signal processing. He is actually member of IEEE, of the Radio-Electronics Vietnam Society (REV).

\end{biography}

\end{document}